\documentclass[12pt]{article}

\usepackage{amsmath,amsfonts,amssymb,amsthm,cite}

\def\N{\mathbb{N}}

\def\e{\varepsilon}
\def\d{\delta}

\def\1{\mathds{1}}

\def\M{\mathcal{M}}
\def\P{\mathcal{P}}
\def\Q{\mathcal{Q}}

\def\B{\mathcal{B}}

\def\MA{M(X,\Sigma)}
\def\M1{M_1(X,\Sigma)}
\def\V{W}
\def\m{\mathfrak{m}}

\def\Ha{\mathop{H_\alpha}}
\def\ha{\mathop{h_\alpha}}
\def\H{\mathop{H}}
\def\h{\mathop{h}}
\def\hva{\mathop{h_{\alpha}^W}}
\def\Hva{\mathop{H_{\alpha}^W}}
\def\hv{\mathop{h^W}}

\def\ldim{\underline{\mathrm{dim}}}
\def\udim{\overline{\mathrm{dim}}}

\def\dim{\mathrm{dim}}

\def\for{\mbox{  for }}

\newtheorem{observation}{Observation}[section]
\newtheorem{theo}{Theorem}[section]
\newtheorem{prop}{Proposition}[section]

\newtheorem{cor}{Corollary}[section]

\theoremstyle{definition}
\newtheorem{definition}{Definition}[section]

\newtheorem{example}{Example}[section]
\newtheorem{problem}{Problem}[section]

\title{Weighted Approach to R\'enyi Entropy}
\author{Marek \'{S}mieja \hspace{2cm} Jacek Tabor}
\date{Preprint}
\begin{document}

\maketitle

\begin{abstract}
R\'enyi entropy of order $\alpha$ is a general measure of entropy. In this paper  we derive estimations for the R\'enyi entropy of the mixture of sources in terms of the entropy of the single sources. These relations allow to compute the R\'enyi entropy dimension of arbitrary order of a mixture of measures.

The key for obtaining these results is our new definition of the weighted R\'enyi entropy. It is shown that weighted entropy is equal to the classical R\'enyi entropy.
\end{abstract}

\section{Introduction.}
The investigation of data compression, coding and behaviour of statistical and physical systems involves the notion of the entropy \cite{El, Go, Se, Wu}. In information theory it is a basic term which is interpreted as a limit of best possible lossless compression of any communication. The reader interested in the history of information theory and data compression is referred to \cite{Be, Ve}.

R\'enyi entropy of order $\alpha$ \cite{Re2, Re1} is a kind of measure of entropy which extends the notion of classical Shannon entropy \cite{Sh}. The advantage of the R\'enyi entropy over the Shannon entropy lies in its generality. When modifying the parameter $\alpha$ it is possible to emphasise or weaken the relevance of some probability events \cite{Ca}. Moreover, contrary to Shannon entropy, there exists efficient methods for computing R\'enyi entropy for some values of parameter $\alpha$ \cite{Gr}. 

The aim of this paper is to adapt the idea and results obtained recently for Shannon entropy \cite{Sm} on the field of R\'enyi entropy of order $\alpha \in (0,\infty) \setminus \{1\}$. The paper \cite{Sm} provides the weighted approach to the Shannon entropy which is based on measures instead of partitions. Weighted entropy describes the idea of random lossy coding which, from a practical point of view, can be more important than the classical deterministic coding.\footnote{Roughly speaking, random coding allows to code a given element once for one symbol and another time for other one. It is not necessary then to control precisely the way of coding.} It is shown that this approach is equivalent to the classical one \cite[Theorem II.1]{Sm}. It occurs that this alternative definition allows to derive some property concerning the Shannon entropy relatively easily. In particular, we obtain the estimation of Shannon entropy of the mixture of sources \cite[Theorem III.1]{Sm} and entropy dimension of convex combination of measures \cite[Theorem IV.1]{Sm}.

In our paper we show that the weighted entropy can be similarly defined for the R\'enyi entropy of any order $\alpha \in (0,\infty) \setminus\{1\}$ (precise definitions will be given in Section \ref{definicje}). As a main result we prove the equivalence between the definition of weighted and classical R\'enyi entropy (see Theorem \ref{wnWaz}). Consequently, we derive an optimal estimation of the R\'enyi entropy and R\'enyi entropy dimension of order $\alpha$ of the mixture of sources (see Theorem \ref{corEnt} and \ref{corDim}).

\section{Basic definitions and their interpretations.}

In this paper, if not stated otherwise, we always assume that $(X, \Sigma, \mu)$ is a probability space.

\subsection{Weighted Shannon entropy.}

To understand weighted approach to entropy, let us recall the basic concepts, definitions and facts from \cite{Sm} where this idea is introduced for Shannon entropy. We use some of these definitions in the present paper. 

The problem of lossy data compression involves transformation between possibly uncountable set of data and some countable coding set. This coding mapping is usually obtained by defining a partition of a data set $X$. We say that a family $\P \subset \Sigma$ is a $\mu$-partition of $X$ if $\P$ is countable family of disjoint sets and
$$
\mu(X \setminus \bigcup_{P \in \P}P) = 0.
$$
Once the partition is chosen, the coding is defined precisely: we code a given $x \in X$ by $P \in \P$ iff $x \in P$. Thus the partition is also called the coding alphabet. 

To obtain a statistical amount of memory per one element used in the lossy coding generated by partition $\P$ we use the Shannon $\mu$-entropy of $\P$ defined by
$$
\h(\mu;\P):= - \sum_{P \in \P} \mu(P) \log_2(\mu(P)).
$$

The coding with use of a given partition causes specific level of error. To control the maximal error we make in the lossy coding we choose the error-control family $\Q$ which is simply the family of measurable subsets of $X$. We consider only such partitions $\P$ which are finner than $\Q$ i.e. we desire that for every $P \in \P$ there exists $Q \in \Q$ such that $P \subset Q$. If this is the case then we say that $\P$ is $\Q$-acceptable and we write $\P \prec \Q$. 

Consequently, to describe the best lossy coding determined by $\Q$-acceptable alphabets we define Shannon $\mu$-entropy of $\Q$ by:
$$
\H(\mu;\Q):=\inf\{\h(\mu;\P) \in [0,\infty] : \mbox{$\P$ is a $\mu$-partition and $\P \prec \Q$}\}.
$$
Similar notions of the entropy in special cases were used by A. R\'enyi \cite{Re2, Re1} and E. C. Posner \cite{Po2, Po3, Po1}. They rather considered error-control families in metric spaces consisted of balls with given radius or cubes with specific edge length.

The inspiration of weighted entropy, lies in the substitution of the division of space $X$ into partition by the division of measure $\mu$ into ``submeasures''\footnote{The idea of weighted entropy is indebted to the notion of weighted Hausdorff measures considered by J. Howroyd \cite{Ho1, Ho2}.}. It enables to use of the operations on functions rather than on plain sets. Roughly speaking, this approach provides the computation and interpretation of the entropy with respect to ``formal'' convex combination $a_1\P_1+a_2\P_2$, where $\P_1,\P_2$ are partitions (which clearly does not make sense in the classical approach).

Let us denote the division of measure $\mu$ with respect to $\Q \subset \Sigma$ by:
\begin{equation} \label{ZbiorV}
\begin{array}{ll}
\V(\mu;\Q):= & \{\m:\Q \ni Q \rightarrow \m_Q \in \MA : \\[0.4ex]
& \m_Q(X \setminus Q) = 0 \text{ for every $Q \in \Q$ and } \sum_{Q \in \Q} \m_Q = \mu \},
\end{array}
\end{equation}
where $\MA$ is the family of all measures on $(X, \Sigma)$. Observe that every function $\m \in \V(\mu;\Q)$ is non-zero on at most countable sets of $\Q$. Then we define the weighted $\mu$-entropy of a given $\m \in \V(\mu;\Q)$ by:
$$
\hv(\mu; \m):= - \sum_{Q \in \Q} \m_Q(X) \log_2(\m_Q(X)) \text{.}
$$
The sum in the above formula is taken over $Q \in \Q$ such that $\m_Q(X) > 0$. The weighted entropy also can be generalised for arbitrary family $\Q \subset \Sigma$, taking the infimum over all functions of $\V(\mu; \Q)$.

It is shown \cite[Theorem II.1]{Sm} that the weighted $\mu$-entropy of $\Q$ is equal to the classical one i.e. $\H(\mu;\Q)$. This equality allows to compute the entropy of the mixture of sources. Let us recall the motivation of this problem:

\begin{problem} \label{probOfMix}
Suppose that we are given two sources $S_1$, $S_2$, which are represented by probability measures $\mu_1$, $\mu_2$ respectively. We assume that the error-control family $\Q$ defines the precision in the lossy-coding elements of $X$. Let us consider a new source $S$ which sends a signal produced by $S_1$ with probability $a_1$ and produced by $S_2$ with probability $a_2$. We are interested in estimation of $\H(a_1\mu_1 + a_2\mu_2;\Q)$ (mixture of $S_1$ and $S_2$) in terms of $\H(\mu_1;\Q)$ and $\H(\mu_2;\Q)$.
\end{problem} 

The following theorem gives the exact solution how to estimate the entropy of the mixture:

\medskip
\begin{flushleft} Shannon entropy of the mixture\cite[Theorem III.1]{Sm}: \emph{Let $a_1, a_2 \in [0,1]$ be such that $a_1 + a_2 = 1$. If $\mu_1, \mu_2$ are probability measures and $\Q \subset \Sigma$ then:}
$$
\H(a_1\mu_1+a_2\mu_2;\Q) \geq a_1 \H(\mu_1;\Q)+a_2\H(\mu_2;\Q)  
$$
\emph{and}
$$
\H(a_1\mu_1+a_2\mu_2;\Q) \leq a_1 \H(\mu_1;\Q)+a_2\H(\mu_2;\Q) - a_1 \log_2(a_1) - a_2 \log_2(a_2).
$$
\medskip
\end{flushleft}

\subsection{Weighted approach to the R\'enyi entropy of order $\alpha$.} \label{definicje}

In further parts of the paper we investigate the weighted approach to the R\'enyi entropy of order $\alpha$. In this subsection we define precisely the weighted R\'enyi entropy. We use the idea from \cite{Sm} described in the previous subsection.

Before that, we recall the classical definition of the R\'enyi entropy of order $\alpha$ for $\mu$-partition and its generalisation for family of measurable subsets of $X$.

\begin{definition} \label{entDef1} 
Let $\alpha \in (0,\infty)\setminus\{1\}$. Given a $\mu$-partition $\P \subset \Sigma$ of $X$, \emph{R\'enyi $\mu$-entropy of order $\alpha$ of $\P$} \cite{Re1} is defined by
$$
{\ha}(\mu;\P):=\frac{1}{1-\alpha} \log_2[\sum_{P \in \P} \mu(P)^\alpha].
$$
For $\Q \subset \Sigma$ we define \emph{R\'enyi $\mu$-entropy of order $\alpha$ of $\Q$} by
$$
{\Ha}(\mu;\Q):=
\inf\{{\ha}(\mu;\P) \in [0,\infty] \, : \, 
\mbox{$\P$ is a $\mu$-partition and $\P \prec \Q$}\}.
$$
\end{definition}

Observe that if there is no $\mu$-partition finer than $\Q$ then
directly from the definition\footnote{We put $\inf(\emptyset) = \infty$.} ${\Ha}(\mu;\Q)=\infty$. Moreover, if $\Q$ itself is a $\mu$-partition of $X$ then trivially ${\Ha}(\mu;\Q)={\ha}(\mu;\Q)$.

As it was mentioned in the previous subsection, the partition describes the way of coding elements of $X$ by the elements of some countable coding set. Given the maximal error we are allowed to make in the process of lossy coding, represented by the measurable family $\Q$ of $X$, we consider all $\Q$-acceptable partitions and choose the one which provides the lowest entropy.

To see that ${\Ha}(\mu;\Q)$ does not have to be attained it is sufficient to use the trivial example from \cite[Example II.1]{Sm}:
\begin{example}
Let $X=(0,1)$, $\Sigma$ be a sigma algebra generated by all Borel subsets of $(0,1)$, $\mu$ be a Lebesgue measure and $\Q$ be an error-control family defined by 
$$
\Q = \{[a,b]: 0 < a < b < 1\}.
$$ 
Clearly ${\Ha}(\mu;\Q) = 0$ but for every $\mu$-partition $\P \prec \Q$, we have ${\Ha}(\mu;\Q) > 0$ when $\alpha \in (0,\infty)\setminus\{1\}$.
\end{example}

Inspired by the reasoning used in \cite{Sm} we construct a definition of weighted R\'enyi $\mu$-entropy of order $\alpha$. The form of set $\V(\mu;\Q)$ -- a division of a measure $\mu$ with respect to error-control family $\Q$, remains the same as in formula (\ref{ZbiorV}). The function form $\V(\mu;\Q)$ defines the set of measures. These measures determine the probability of encoding the given element with the code represented by the specific measure. Since every element $x \in X$ can be encoded once with one code and second time by another one, the coding is called random.

\begin{definition}
Let $\Q \subset \Sigma$ and let $\alpha \in (0,\infty)\setminus\{1\}$. We define the {\em weighted R\'enyi $\mu$-entropy} of order $\alpha$ of $\m \in \V(\mu;\Q)$ by
\begin{equation} \label{defWeiRen}
{\hva}(\mu; \m):=\frac{1}{1-\alpha}\log_2(\sum_{Q \in \Q} \m_Q(X)^\alpha) \text{.}
\end{equation}
The weighted R\'enyi $\mu$-entropy of order $\alpha$ of $\Q$ is defined by
$$
{\Hva}(\mu; \Q) := \inf \{ {\hva}(\mu; \m) \in [0,\infty] : \m \in \V(\mu,\Q) \} \text{.}
$$
\end{definition}
The sum in formula (\ref{defWeiRen}) is taken over $Q \in \Q$ such that $\m_Q(X) > 0$, consequently this is a countable sum.

The above definitions allow to perform the operations on functions when computing the entropy. Such operations will be crucial to derive some estimations of R\'enyi entropy and entropy dimension of the mixture of sources shown in the following sections. To use all the advantages of weighted R\'enyi entropy it remains to show the equivalence between weighted and classical entropy. It is made in the next section.

We now make some additional notations and observations which we will refer to very often in future. We denote by $g_\alpha$ and its inverse $g^{-1}_\alpha$ the following functions:
$$
g_\alpha(x) = 2^{(1-\alpha)x} \text{, } g^{-1}_\alpha(x) = \frac{1}{1-\alpha}\log_2(x).
$$

Then the observation is valid:
\begin{observation} \label{obserwacja} 
\begin{enumerate}
\item If $\alpha \in (0,1)$ then:
\begin{enumerate}
\item $g^{-1}_\alpha$ and $g_\alpha$ are ascending,\label{asc}
\item $x \rightarrow x^\alpha$ is subadditive,\label{sub}
\item $x \rightarrow x^\alpha$ is concave \label{conc}
\end{enumerate} 
\item If $\alpha \in (1,\infty)$ then:
\begin{enumerate}
\item $g^{-1}_\alpha$ and $g_\alpha$ are descending,\label{des}
\item $x \rightarrow x^\alpha$ is superadditive,\label{sup}
\item $x \rightarrow x^\alpha$ is convex.\label{conv}
\end{enumerate}
\end{enumerate}
\end{observation}


\section{Equivalence between classical and weighted R\'enyi entropy of order $\alpha$.}

The purpose of this section is to show that weighted R\'enyi entropy with respect to the family $\Q \subset \Sigma$ equals the classical R\'enyi entropy of $\Q$. It will allow us to use benefits of this alternative definition in further analysis. In the proofs we apply the idea used in \cite{Sm}.

The equality will be derived in two steps. First we show the inequality ${\Hva}(\mu; \Q) \leq {\Ha}(\mu; \Q)$. The inequality can be interpreted as a deterministic coding is a special case of specific random one. More difficult is to show the opposite inequality. It involves the application of Hardy-Littlewood-Polya Theorem. 

We start with first inequality:
\begin{prop} \label{prop}
Let $\alpha \in (0,\infty)\setminus\{1\}$. Then
$$
{\Hva}(\mu; \Q) \leq {\Ha}(\mu; \Q),
$$
for every family $\Q \subset \Sigma$.
\end{prop}

\begin{proof}
Let us first observe that if there is no $\mu$-partition finer than $\Q$ then ${\Ha}(\mu;\Q)=\infty$ and the inequality holds trivially.

Thus let us assume that it is not the case. Let $\P$ be a $\mu$-partition finer than $\Q$. Our aim is to construct a function $\m \in \V(\mu;\Q)$ with lower entropy than $\P$. 

First, since $\P \prec \Q$, then we obtain a mapping $\pi: \P \to \Q$ such that $P \subset \pi(P)$. Next we put 
$$
\P_\Q := \{P_Q\}_{Q \in \Q},
$$
where $P_Q:=\bigcup\limits_{P:\pi(P)=Q}P$. Finally, we define $\m:\Q \ni Q \rightarrow \mu_{|P_Q} \in \MA$.

We verify that $\m \in \V(\mu; \Q)$. Clearly, $\P_\Q$ is a $\mu$-partition and $P_Q \subset Q$ for every $Q \in \Q$. Hence
$$
\sum_{Q \in \Q} \m_Q(X) = \sum_{Q \in \Q} \mu_{|P_Q}(Q) 
= \sum_{Q \in \Q} \mu(P_Q) = \mu(X).
$$
The above sums are taken only over $Q \in \Q$ such that $\m_Q(X) > 0$. Moreover, we have
$$
\m_Q(X \setminus Q) = \mu_{|P_Q}(X \setminus Q) 
\leq \mu_{|Q}(X \setminus Q)= 0,
$$
for $Q \in \Q$. We obtain that $\m \in \V(\mu;\Q)$. 

It remains to check that $\hva(\mu;\m) \leq \ha(\mu;\P)$. To see this we use Observation \ref{obserwacja}: \ref{asc} and \ref{sub} for $\alpha \in (0,1)$ or \ref{des} and \ref{sup} for $\alpha \in (1,\infty)$. More precisely, we have:
$$
{\hva}(\mu; \m) 
= \frac{1}{1-\alpha}\log_2\big(\sum_{Q \in \Q}\m_Q(X)^\alpha\big) 
$$
$$
= \frac{1}{1-\alpha}\log_2\big(\sum_{Q \in \Q}\mu_{|P_Q}(X)^\alpha\big) 
= \frac{1}{1-\alpha}\log_2\big(\sum_{Q \in \Q}\mu(P_Q)^\alpha\big) 
$$
$$
= \frac{1}{1-\alpha}\log_2\big(\sum_{Q \in \Q}\mu(\bigcup_{P: \pi(P)=Q}P)^\alpha\big) 
\leq \frac{1}{1-\alpha}\log_2\big(\sum_{Q \in \Q} \sum_{P: \pi(P)=Q}\mu(P)^\alpha\big) 
$$
$$
= \frac{1}{1-\alpha}\log_2\big(\sum_{P \in \P} \mu(P)^\alpha\big)
={\ha}(\mu; \P).
$$
As $P$ was chosen as arbitrary partition, we conclude that ${\Hva}(\mu; \Q) \leq {\Ha}(\mu; \Q)$ for $\alpha \in (0,\infty)\setminus\{1\}$.
\end{proof}

As it was mentioned to derive the inequality ${\Hva}(\mu;\Q) \geq {\Ha}(\mu;\Q)
$ it is necessary to use Hardy Littlewood Polya Theorem. The version of Hardy Littlewood Polya Theorem for finite sequences is given in \cite[Theorem 1.5.4]{Ni}. Its generalisation for infinite sequences can be relatively easily achieved (see \cite[Appendix A]{Sm}). Let us recall this Theorem:
\medskip
\begin{flushleft} Hardy Littlewood Polya Theorem. \emph{Let $a>0$ and let $\varphi:[0,a] \to (0,\infty)$, $\varphi(0)=0$ be a continuous function.
Let $(x_i)_{i \in I}, (y_i)_{i \in I} \subset [0,a]$ be given sequences where
either $I=\N$ or $I=\{1,\ldots,N\}$ for a certain $N \in \N$. We assume that $(x_i)_{i \in I}$ is a nonincreasing sequence and}
$$
\sum_{i=1}^n x_i \leq \sum_{i=1}^n y_i \for n \in I,
$$
$$
\sum_{i \in I}x_i=\sum_{i \in I}y_i.
$$
\emph{Then}
\begin{itemize}
\item $\sum_{i \in I} \varphi(x_i) \geq \sum_{i \in I} \varphi(y_j)$ if $\varphi$ is concave,
\item $\sum_{i \in I} \varphi(x_i) \leq \sum_{i \in I} \varphi(y_j)$ if $\varphi$ is convex.
\end{itemize}
\medskip
\end{flushleft}

We first show an additional proposition:

\begin{prop} \label{waz}
Let $\Q =\{Q_i\}_{i \in I}$ be a family of measurable subsets of $X$, where either $I=\N$ or $I=\{1,\ldots,N\}$ for a certain $N \in \N$. Let $\m \in \V(\mu;\Q)$. We assume that
\begin{itemize}
 \item $\mu(X \setminus \bigcup\limits_{i \in I}Q_i) = 0$,
 \item the sequence $I \ni i \rightarrow \m_{Q_i}(X)$ is nonincreasing.
\end{itemize}
We define the family $\P=\{P_i\}_{i \in I} \subset \Sigma$ by the formula
$$
P_1:=Q_1, \, P_i:=Q_i \setminus \bigcup_{k=1}^{i-1} Q_{k} \for i \in I, i \geq 2.
$$
Then $\P$ is a $\mu$-partition, $\P \prec \Q$ and
\begin{equation} \label{cosik}
{\hva}(\mu; \m) \geq {\ha}(\mu;\P)
\end{equation}
for $\alpha \in (0,\infty)\setminus\{1\}$.
\end{prop}

\begin{proof}
By the definition of family $\P$, we get that $\P \prec \Q$. Moreover, $\P$ is a $\mu$-partition since $\mu(X \setminus \bigcup\limits_{i \in I}Q_i)=0$ and $\bigcup\limits_{i \in I} P_i = \bigcup\limits_{i \in I} Q_i$.

To prove (\ref{cosik}) we use Hardy Littlewood Polya Theorem. The sequences $(x_i)_{i \in I} \subset [0,1]$ and $(y_i)_{i \in I} \subset [0,1]$ are defined by the formulas
$$
x_i:=\m_{Q_i}(X)=\m_{Q_i}(Q_i), \, 
y_i:=\mu(P_i)
$$
for $i \in I$.  

Directly from the assumption we get that $(x_i)_{i \in I}$ is nonincreasing and
$$
\sum_{i \in I} x_i=\mu(X)=\sum_{i \in I} y_i.
$$
Moreover, for every $n \in I$:
$$
\sum_{i = 1}^n x_i=\sum_{i = 1}^n \m_{Q_i}(Q_i) 
=(\sum_{i = 1}^n \m_{Q_i})(Q_1 \cup \ldots \cup Q_n)
$$
$$
\leq \mu(Q_1 \cup \ldots \cup Q_n)
=\sum_{i = 1}^n \mu(P_i)=\sum_{i = 1}^n y_i.
$$

We verified that sequences satisfy assumptions of Hardy Littlewood Polya Theorem. Thus given a function $x \rightarrow x^\alpha$ (we also use Observation \ref{obserwacja}: \ref{asc} and \ref{conc} for $\alpha \in (0,1)$ or \ref{des} and \ref{conv} for $\alpha \in (1,\infty)$) we conclude that
$$
{\hva}(\mu; \m) = \frac{1}{1-\alpha}\log_2\big(\sum_{i \in I} \m_{Q_i}(X)^\alpha\big) 
= \frac{1}{1-\alpha}\log_2(\sum_{i \in I}x_i^\alpha) 
$$
$$
\geq \frac{1}{1-\alpha}\log_2(\sum_{i \in I}y_i^\alpha) 
= \frac{1}{1-\alpha}\log_2\big(\sum_{i \in I}\mu(P_i)^\alpha\big)
={\ha}(\mu;\P),
$$
which completes the proof.
\end{proof}

We are now ready to formulate and complete the proof of the equivalence between classical and weighted R\'enyi $\mu$-entropy of order $\alpha$.

\begin{theo} \label{wnWaz}
Let $\Q \subset \Sigma$. Then weighted R\'enyi $\mu$-entropy coincides with the classical R\'enyi $\mu$-entropy, i.e.
$$
{\Hva}(\mu;\Q) = {\Ha}(\mu;\Q)
$$
for $\alpha \in (0,\infty)\setminus\{1\}$.
\end{theo}

\begin{proof}
It is sufficient to show that ${\Hva}(\mu;\Q) \geq {\Ha}(\mu;\Q)$ since the opposite inequality follows directly from Proposition \ref{prop}.

Let us first observe that if $\V(\mu;\Q) = \emptyset$ then ${\Hva}(\mu;\Q) = \infty$ and trivially ${\Hva}(\mu;\Q) \geq {\Ha}(\mu;\Q)$. 

We discuss the case when $\V(\mu;\Q) \neq \emptyset$. Let $\m \in \V(\mu;\Q)$ be an arbitrary function. We define the subset of family $\Q$ by:
$$
\tilde{\Q} := \{Q \in \Q : \m_Q(X) > 0\}.
$$
Let us notice that $\tilde{\Q}$ is a countable family since $\sum\limits_{Q \in \tilde{\Q}} \m_Q(X) = 1$. Clearly, $\tilde{\m}:=\m_{|\tilde{\Q}} \in \V(\mu; \tilde{\Q})$. Moreover, $\tilde{\Q} \prec \Q$ and ${\hva}(\mu;\tilde{\m})={\hva}(\mu;\m)$.

As $\tilde{\Q}$ is countable, we may find a set of indices $I \subset \N$ such that $\tilde{\Q}=\{Q_i\}_{i \in I}$ and the sequence $I \ni i \rightarrow \m_{Q_i}(X)$ is nonincreasing. Making use of Proposition \ref{waz} we construct a $\mu$-partition $\P \prec \tilde{\Q}$, which satisfies
$$
{\hva}(\mu;\tilde{\m}) \geq {\ha}(\mu;\P).
$$
This completes the proof since $\P \prec \tilde{\Q} \prec \Q$ and ${\hva}(\mu;\m) = {\hva}(\mu;\tilde{\m}) \geq {\ha}(\mu;\P)$.
\end{proof}

Since we show the equality between classical and weighted R\'enyi entropy we use one symbol ${\Ha}(\mu;\Q)$ to denote the R\'enyi entropy of order $\alpha$ with respect to measurable family $\Q$ of $X$.


\section{R\'enyi entropy of order $\alpha$ of the mixture of sources.}

In this section we will show how to apply the definition of weighted R\'enyi entropy to estimate the R\'enyi entropy of the mixture of sources (see Problem \ref{probOfMix}).

Let us start with the proposition:

\begin{prop} \label{propozycja}
We assume that $\alpha \in (0,\infty)\setminus\{1\}$ and $n \in \N$. Let $a_k \in (0,1)$ for $k \in \{1,\ldots,n\}$ be such that $\sum\limits_{k=1}^n a_k = 1$ and let $\{\mu_k\}_{k=1}^n \subset \M1$. We define $\mu := \sum\limits_{k=1}^n a_k \mu_k \in \M1$. 
\begin{itemize}
\item
If $\P$ is a $\mu$-partition of $X$ then $\P$ is a $\mu_k$-partition of $X$ for $k \in \{1,\ldots,n\}$ and
\begin{equation} \label{pierwsza}
{\ha}(\mu; \P) \geq g^{-1}\big[\sum_{k = 1}^n a_k g({\ha}(\mu_k; \P))\big].
\end{equation}
\item If $\Q \subset \Sigma$ and $\m^k \in \V(\mu_k;\Q)$ for $k \in \{1,\ldots,n\}$ then $\m := \sum\limits_{k=1}^n a_k \m^k \in \V(\mu;\Q)$ and
\begin{equation} \label{druga}
{\hva}(\mu; \m) \leq g^{-1}\big[\sum_{k = 1}^n a_k^\alpha g({\hva}(\mu_k; \m^k))\big].
\end{equation}
\end{itemize}
\end{prop}

\begin{proof}
It is easy to see that $\P$ is a $\mu_k$-partition of $X$ for every $k \in \{1,\ldots,n\}$. 

Making use of Observation \ref{obserwacja}: \ref{asc} and \ref{conc} for $\alpha \in (0,1)$ or \ref{des} and \ref{conv} for $\alpha \in (1,\infty)$, we have
$$
{\ha}(\mu;\P)= \frac{1}{1-\alpha} \log_2\big[\sum_{P \in \P}\big(\sum_{k=1}^n a_k\mu_k(P)\big)^\alpha\big]
$$
$$
\geq \frac{1}{1-\alpha} \log_2\big[\sum_{k=1}^n \big(a_k \sum_{P \in \P} \mu_k(P)^\alpha\big)\big]  
$$
$$
=\frac{1}{1-\alpha} \log_2\big[\sum_{k=1}^n a_k 2^{(1-\alpha) \ha(\mu_k;\P)}\big] 
$$
$$
=g^{-1}\big[\sum_{k = 1}^n a_k g({\ha}(\mu_k; \P))\big],	
$$
which proves (\ref{pierwsza}).

We derive the second part of the Proposition. Clearly, $\m \in \V(\mu;\Q)$. To verify (\ref{druga}) we use Observation \ref{obserwacja}: \ref{asc} and \ref{sub} for $\alpha \in (0,1)$ or \ref{des} and \ref{sup} for $\alpha \in (1,\infty)$:
$$
{\hva}(\mu;\m)= \frac{1}{1-\alpha} \log_2\big[\sum_{Q \in \Q}\big(\sum_{k=1}^n a_k\m_Q^k(X)\big)^\alpha\big]
$$
$$
\leq \frac{1}{1-\alpha} \log_2\big[\sum_{k=1}^n \big(a_k^\alpha \sum_{Q \in \Q} \m_Q^k(X)^\alpha\big)\big]  
$$
$$
=\frac{1}{1-\alpha} \log_2\big[\sum_{k=1}^n a_k^\alpha 2^{(1-\alpha) \hva(\mu_k;\m^k)}\big] 
$$
$$
=g^{-1}\big[\sum_{k = 1}^n a_k^\alpha g({\hva}(\mu_k; \m^k))\big].
$$
\end{proof}

Below we present the main theorem in this section concerning the entropy of the mixture of sources. To see that this estimation is sharp we refer the reader to Example \ref{exSharp}.
\begin{theo} \label{corEnt}
Let $\alpha \in (0,\infty)\setminus\{1\}$ and $n \in \N$. We assume that $a_k \in [0,1]$ for $k \in \{1,\ldots,n\}$ be such that $\sum\limits_{k=1}^n a_k = 1$. Let $\{\mu_k\}_{k=1}^n \subset \M1$. If $\Q \subset \Sigma$ then
\begin{equation} \label{nierowWaz}
{\Ha}(\mu; \Q) \geq g^{-1}\big[\sum_{k = 1}^n a_k g({\Ha}(\mu_k; \Q))\big]
\end{equation}
and
\begin{equation} \label{nierowWaz2}
{\Ha}(\mu; \Q) \leq g^{-1}\big[\sum_{k = 1}^n a_k^\alpha g({\Ha}(\mu_k; \Q))\big].
\end{equation}
\end{theo}

\begin{proof}
Let us first consider the case when ${\Ha}(\mu_k;\Q) = \infty$ for a certain $k\in\{1,\ldots,n\}$. Then also ${\Ha}(\mu;\Q) = \infty$ and the inequalities hold trivially. 

Thus let us assume that for every $k\in\{1,\ldots,n\}$, ${\Ha}(\mu_k;\Q) < \infty$. Without loss of generality, we may assume also that $a_k \neq 0$ for every $k \in \{1,\ldots,n\}$. Let $\e>0$ be arbitrary and let $\mu := \sum\limits_{k=1}^n a_k \mu_k$. 

To prove the first inequality, we find a $\mu$-partition $\P$ finer than $\Q$ such that 
\begin{equation} \label{gwiazdka}
{\Ha}(\mu;\Q) \geq {\ha}(\mu;\P)-\e. 
\end{equation}
Consequently, by Proposition \ref{propozycja} and the definition of R\'enyi entropy, we have
$$
{\ha}(\mu;\P) = {\ha}(\sum_{k=1}^n a_k \mu_k; \P) 
$$
$$
\geq g^{-1}\big[\sum_{k = 1}^n a_k g({\ha}(\mu_k; \P))\big] \geq g^{-1}\big[\sum_{k = 1}^n a_k g({\Ha}(\mu_k; \P))\big].
$$
Finally by (\ref{gwiazdka}), we obtain
$$
{\Ha}(\mu;\Q) \geq {\ha}(\mu; \P) -\e 
\geq g^{-1}\big[\sum_{k = 1}^n a_k g({\Ha}(\mu_k; \P))\big] - \e,
$$
which proves (\ref{nierowWaz}).

We prove the inequality (\ref{nierowWaz2}). For each $k \in \{1,\ldots n\}$ we find $\m^k \in \V(\mu_k;\Q)$ satisfying
\begin{equation} \label{2gwiazdki}
{\hva}(\mu_k; \m^k) \leq {\Ha}(\mu_k;\Q) + \frac{\e}{n}.
\end{equation}

Making use of Proposition \ref{propozycja} and (\ref{2gwiazdki}), we have
$$
{\Ha}(\mu;\Q) \leq g^{-1}\big[\sum_{k = 1}^n a_k^\alpha g({\hva}(\mu_k; \m^k))\big] \leq g^{-1}\big[\sum_{k = 1}^n a_k^\alpha g({\Ha}(\mu_k; \m^k))\big] + \e.
$$
This completes the proof as $\e>0$ was an arbitrary number.
\end{proof}

The following examples confirms that the above estimation (\ref{nierowWaz}) and (\ref{nierowWaz2}) cannot be improved.

\begin{example} \label{exSharp}
Let us assume that $X=\{0,1\}$ and $\mu_1, \mu_2$ denote discrete measures such that:
$$
	\mu_1(\{0\})=1 \text{ and } \mu_2(\{1\})=1.
$$
Then, we have 
$$
{\Ha}(a_1\mu_1+a_2\mu_2) = \frac{1}{1-\alpha}\log_2(a_1^\alpha + a_2^\alpha).
$$ 
It is exactly the right side of the inequality (\ref{nierowWaz2}). 

On the other hand, if we consider two measures which satisfy $\mu_1 = \mu_2$, then 
$$
{\Ha}(a_1\mu_1+a_2\mu_2) = {\Ha}(\mu_1)={\Ha}(\mu_2)
$$
and it equals the right side of (\ref{nierowWaz}).
\end{example}

Let us observe the similarity between bounds obtained for both, Shannon entropy \cite[Theorem III.1]{Sm} and R\'enyi entropy of order $\alpha$ from Theorem \ref{corEnt}. Let us consider the functions:
$$
l_\alpha(x,y) = g_\alpha^{-1}(a_1 g_\alpha(x) + a_2 g_\alpha(y)),
$$
$$
u_\alpha(x,y) = g_\alpha^{-1}(a_1^\alpha g_\alpha(x) + a_2^\alpha g_\alpha(y)),
$$
which describe the lower and upper bound for the R\'enyi entropy of order $\alpha$. If $x, y$ are non negative real numbers then these functions converge to the corresponding bounds calculated for Shannon entropy as $\alpha \to 1$ i.e.:
$$
\left\{
   \begin{array}{ll}
		l_\alpha(x,y) \to a_1 x + a_2 y \\ 
		u_\alpha(x,y) \to a_1 x + a_2 y  - a_1 \log_2(a_1) - a_2 \log_2(a_2)
	\end{array}
\text{, when } \alpha \to 1
\right.
$$


\section{R\'enyi entropy dimension of order $\alpha$.}

In this section we show the estimation of R\'enyi entropy dimension of order $\alpha$ of the combination of measures in terms of the entropy dimension of the single measures. Before proceeding with it let us recall the definition of R\'enyi entropy dimension of order $\alpha$. In this section we additionally assume that $X$ is a metric space and $(X, \Sigma, \mu)$ is a probability space, where $\Sigma$ contains all Borel subsets of $X$.

Given $\d>0$ let us denote a family of all balls in $X$ with radius $\d$ by
$$
\B_\d := \{B(x,\d): x \in X\},
$$
where $B(x,\d)$ is a closed ball centred at $x$ with radius $\d$.
\begin{definition}
The \emph{upper and lower R\'enyi entropy dimension of order} $\alpha \in (0,\infty)\setminus\{1\}$ of measure $\mu \in \M1$ are defined by
$$
{\udim}_\alpha(\mu):=\limsup_{\d \to 0} \frac{{\Ha}(\mu;\B_\d)}{-\log_{2}(\d)}, \,
$$
$$
{\ldim}_\alpha(\mu):=\liminf_{\d \to 0} \frac{{\Ha}(\mu;\B_\d)}{-\log_{2}(\d)}.
$$
If the above are equal we say that $\mu$ has the \emph{R\'enyi entropy dimension of order $\alpha$} and denote it by ${\dim}_\alpha(\mu)$.
\end{definition}

The following theorem gives the estimation of the R\'enyi entropy dimension of the mixture of measures.
\begin{theo} \label{corDim}
Let $a_1, a_2 \in (0,1)$ be such that $a_1 + a_2 = 1$ and let $\mu_1, \mu_2 \in \M1$. If ${\udim}_\alpha(\mu_1) < \infty$ and ${\udim}_\alpha(\mu_2) < \infty$ then
\begin{equation} \label{jeden}
{\udim}_\alpha(a_1 \mu_1 + a_2 \mu_2) \leq \left\{
   \begin{array}{ll}
		\max \{\udim_\alpha(\mu_1),\udim_\alpha(\mu_2)\} & \mbox{for } \alpha \in (0,1), \\ 
		\min \{\udim_\alpha(\mu_1),\udim_\alpha(\mu_2)\} & \mbox{for } \alpha \in (1,\infty)
	\end{array}
\right.
\end{equation}
and
\begin{equation} \label{dwa}
{\ldim}_\alpha(a_1 \mu_1 + a_2 \mu_2) \geq \left\{
   \begin{array}{ll}
		\max \{\ldim_\alpha(\mu_1),\ldim_\alpha(\mu_2)\} & \mbox{for } \alpha \in (0,1),\\ 
		\min \{\ldim_\alpha(\mu_1),\ldim_\alpha(\mu_2)\} & \mbox{for } \alpha \in (1,\infty).
	\end{array}
\right.
\end{equation}
\end{theo}
\begin{proof}
Let us show first inequality from formula (\ref{dwa}). Rest of them can be proven in similar manner.

Directly from the definition of R\'enyi entropy dimension of order $\alpha \in (0,\infty)\setminus\{1\}$, we have:
$$
\liminf_{\d \to 0} \frac{\Ha(\mu_k; \B_\d)}{-\log_2(\d)} = \ldim_\alpha(\mu_k) \text{, for } k = 1,2.
$$
Then for arbitrary $\e_1, \e_2 > 0$, there exists $\d_1, \d_2 > 0$, such that:
$$
\frac{\Ha(\mu_k; \B_{\d_k})}{-\log_2(\d_k)} \geq \ldim_\alpha(\mu_k) - \e_k
$$
and consequently
$$
\Ha(\mu_k; \B_{\d_k}) \geq -\log_2(\d_k)(\ldim_\alpha(\mu_k) - \e_k),
$$
for $k=1,2$. 

We put $\d:=\min\{\d_1,\d_2\}$. Making use of Observation \ref{obserwacja}: \ref{asc}, we get:
$$
g_\alpha^{-1}\big\{a_1 g_\alpha\big[\Ha(\mu_1; \B_\d)\big]+a_2 g_\alpha\big[\Ha(\mu_2; \B_\d)\big]\big\} 
$$
$$
\geq g_\alpha^{-1}\big\{a_1 g_\alpha\big[-\log_2(\d)(\ldim_\alpha(\mu_1) - \e_1)\big]+a_2  g_\alpha\big[-\log_2(\d)(\ldim_\alpha(\mu_2) - \e_2)\big]\big\}.
$$
By Theorem \ref{corEnt}, we have:
$$
\Ha(a_1\mu_1+a_2\mu_2;\B_\d) \geq
g_\alpha^{-1}\big\{a_1 g_\alpha\big[\Ha(\mu_1; \B_\d)\big]+a_2 g_\alpha\big[\Ha(\mu_2; \B_\d)\big]\big\} 
$$
$$
\geq g_\alpha^{-1}\big\{a_1 g_\alpha\big[-\log_2(\d)(\ldim_\alpha(\mu_1) - \e_1)\big]+a_2 g_\alpha\big[-\log_2(\d)(\ldim_\alpha(\mu_2) - \e_2)\big]\big\}.
$$
Dividing the above inequality by $(-\log_2(\d))$ and taking the limit as $\d \to 0$, we conclude:
$$
\liminf_{\d \to 0}\frac{\Ha(a_1\mu_1 + a_2\mu_2;\B_\d)}{-\log_2(\d)}
$$
$$
\geq \liminf_{\d \to 0} \frac{g_\alpha^{-1}\big\{a_1 g_\alpha\big[-\log_2(\d)(\ldim_\alpha(\mu_1) - \e_1)\big] + a_2 g_\alpha\big[-\log_2(\d)(\ldim_\alpha(\mu_2) - \e_2)\big]\big\}}{-\log_2(\d)}
$$
$$
=\liminf_{\d \to 0} \frac{\frac{1}{1-\alpha}\log_2\big[a_1 \d^{-(1-\alpha)(\ldim_\alpha(\mu_1) - \e_1)} + a_2 \d^{-(1-\alpha)(\ldim_\alpha(\mu_2) - \e_2)}\big]}{-\log_2(\d)}
$$
$$
=\liminf_{\d \to 0} \frac{\frac{1}{1-\alpha}\log_2\big\{\d^{-(1-\alpha)(\ldim_\alpha(\mu_1) - \e_1)}\big[a_1 + a_2 \d^{(1-\alpha)(\ldim_\alpha(\mu_1) - \ldim_\alpha(\mu_2) - \e_1 + \e_2)}\big]\big\}}{-\log_2(\d)}
$$
$$
=\ldim_\alpha(\mu_1) - \e_1.
$$
Since $\e_1, \e_2$ was the arbitrary numbers, then desired inequality holds.
\end{proof}
Clearly, the above theorem can be generalised for any finite number of measures. In the case when all measures have R\'enyi entropy dimension of order $\alpha$ then the entropy dimension of the convex combination of measures is determined precisely.
\begin{cor}
Let $a_k \in (0,1)$ for $k=1,\ldots,n$ be such that $\sum\limits_{k=1}^n a_k = 1$ where $n \in \N$ and let $\{\mu_k\}_{k=1}^n \subset \M1$. If every $\mu_k$ has finite R\'enyi entropy dimension for $k \in \{1,\ldots,n\}$ then $\sum_{k=1}^n \mu_k$ also have R\'enyi entropy dimension. Moreover,
$$
{\dim}_\alpha(\sum_{k=1}^n a_k \mu_k) = \left\{
   \begin{array}{ll}
		\max\limits_{k=1,\ldots,n}\dim(\mu_k) & \mbox{for } \alpha \in (0,1),\\ 
		\min\limits_{k=1,\ldots,n}\dim(\mu_k) & \mbox{for } \alpha \in (1,\infty).
	\end{array}
\right.
$$
\end{cor}

\bibliographystyle{plain}
\bibliography{entropy}

\end{document}